\documentclass[final]{siamart171218}

\usepackage{lipsum}
\usepackage{amsfonts}
\usepackage{graphicx}
\usepackage{epstopdf}
\usepackage{algorithmic}
\ifpdf
  \DeclareGraphicsExtensions{.eps,.pdf,.png,.jpg}
\else
  \DeclareGraphicsExtensions{.eps}
\fi

\newsiamremark{remark}{Remark}
\newsiamremark{hypothesis}{Hypothesis}
\crefname{hypothesis}{Hypothesis}{Hypotheses}
\newsiamthm{claim}{Claim}

\headers{Stabilizability preserving quotients of non-linear systems}{T. Chingozha, O.T. Nyandoro, and M.A. van Wyk}

\title{Stabilizability preserving quotients of non-linear systems}

\author{Tinashe Chingozha\thanks{School of Electrical and Information Engineering, University of the Witwatersrand, Johannesburg, South Africa 
  (\email{tinashe.chingozha@wits.ac.za}, \email{otis.nyandoro@wits.ac.za}, \email{anton.vanwyk@wits.ac.za}.}
\and Otis T. Nyandoro\footnotemark[1]
\and Anton van Wyk\footnotemark[1]}

\usepackage{amsopn}

\ifpdf
\hypersetup{
  pdftitle={Stabilizability preserving quotients of non-linear systems},
  pdfauthor={T. Chingozha, O.T. Nyandoro, and M.A. van Wyk}
}
\fi

\usepackage{tikz-cd}
\usepackage{etoolbox}
\let\bbordermatrix\bordermatrix
\patchcmd{\bbordermatrix}{8.75}{4.75}{}{}
\patchcmd{\bbordermatrix}{\left(}{\left[}{}{}
\patchcmd{\bbordermatrix}{\right)}{\right]}{}{}

\newsiamthm{Defn}{Definition}
\newsiamthm{Prob}{Problem}
\newsiamremark{Remark}{Remark}

\begin{document}

\maketitle

\begin{abstract}
In this paper quotients of control systems which are generalizations of system reductions are used to study the stabilizability property of non-linear systems. Given a control system and its quotient we study under what conditions stabilizability of the quotient is sufficient to guarantee stabilizability of the original system. We develop a novel method of constructing a control Lyapunov function for the original system from the implied Lyapunov function of the quotient system, this construction involves the solution of a system of partial differential equations. By studying the integrability conditions of this associated system of partial differential equations we are able to characterize obstructions to our proposed method of constructing control Lyapunov functions in terms of the structure of the original control system.  
\end{abstract}


\section{Introduction}
Design of stabilizing feedback controllers is a quintessential task of modern control engineering with most of the design methods relying on Lyapunov's stability theory\cite{Kokotovic}. The ability to design a stabilizing feedback controller is therefore tied to the existence of a control Lyapunov function. In the case of lower order systems it is possible to proceed heuristically in the construction of Lyapunov functions however as the dimension increases constructing Lyapunov functions becomes more of an art than a science. To circumvent this ``dimensionality curse'' design methods such as backstepping control, immersion and invariance, sliding mode control e.t.c\cite{AstolfiIandI},\cite{Kokotovic},\cite{Slotine} allow for a hierarchical procedure in designing stabilizing feedback controllers. In this paper we study the question of control system stabilizability from the perspective of control system quotients with the aim being to reveal structural obstructions to stabilization via hierarchical methods.

The notion of quotient control systems as used in this paper was developed in the seminal work of Tabauda et.al\cite{Tabauda} in a category theoretic setting. Consider the category of control systems where the objects of this category are control systems, the morphisms in this category are such that trajectories are mapped between objects. A control system $\tilde{\Sigma}$ is said to be a quotient of $\Sigma$ if it is of a lower order and if there exists a morphism that maps trajectories of $\Sigma$ to trajectories of $\tilde{\Sigma}$. Quotients provide a general framework to describe model order reduction for control systems, reduction methods such as Lie symmetry methods\cite{Krener}, controlled invariant distribution methods\cite{Nijmeijer} and principal fibre bundle methods\cite{Martin} can all be subsumed under the framework of quotients. 

Applying quotients and related reductions techniques to the study of control systems properties is an approach that has been successfully applied with regards to control system controllability. One of the earliest results of this approach is the work by Martin et.al\cite{Martin} where it is shown that for accessible systems modelled on principal fibre bundles with compact structure group controllability of the projection of the system onto the base manifold is necessary and sufficient for the controllability of the original system. Using the related ideas of simulation relations of control systems the propagation of the controllability property by simulation relations is studied in \cite{Grasse} for systems with input disturbances. Given a controllable system $\Sigma$ which is in a simulation relation with another control system $\tilde{\Sigma}$ the authors in \cite{Grasse} give conditions that the simulation relation has to satisfy such that $\tilde{\Sigma}$ is controllable if and only if $\Sigma$ is controllable. For hierarchical controllability the reverse direction(i.e if $\tilde{\Sigma}$ is controllable and simulates the system $\Sigma$ under what condition is $\Sigma$ controllable) is more important since $\tilde{\Sigma}$ can be of lower dimension. It is this direction of inquiry that is pursued in the work of Pappas et.al\cite{PappasConsAbst} using the language of control system abstractions. Given a control system $\Sigma$ defined on a manifold $M$ and an abstraction map $\Phi : M \mapsto N$, a method of constructing a control system on $N$ which is $\Phi$-related to $\Sigma$ is developed in \cite{PappasConsAbst}. Furthermore they prove that if the kernel of the abstraction map is contained in the Lie algebra generated by $\Sigma$'s control bundle then the control system on $N$ is a consistent abstraction of $\Sigma$. Consistency here implies that the abstracted control system is controllable if and only if $\Sigma$ is controllable. These results are used in \cite{PappasHierConsSys} to develop a hierarchical controllability algorithm for linear control systems.

The success of quotients and related ideas in the analysis of controllability unfortunately does not carry over to the equally important property of stabilizability. A complete characterization of stabilizability preserving quotients for linear systems was developed in \cite{PappasLin} where they leverage off the equivalence of complete controllability and linear stabilizability for linear systems and the results in \cite{PappasConsAbst} on controllability preserving abstractions. For non-linear systems however the question is far from being conclusively answered, the major reason being that there is no simple relationship between stabilizability and controllability. This makes it impossible to follow the approach taken by Pappas et.al\cite{PappasLin} for linear systems to leverage off the vast results of controllability preserving quotients since controllability is not a necessary condition for the existence of a smooth stabilizing feedback control\cite{BrockettAsymStab}. It is the aim of this paper then to explicitly address the question of stabilizability preserving quotients by following a Lyapunov theory approach.

This paper is organised as follows. We start with a presentation of the necessary mathematical notations and machinery from differential theory, the geometric theory of partial differential equations and geometric control theory in \cref{sec:MathsPrelim}. \Cref{sec:ProbStat} contains a formal mathematical statement of the problem of stabilizability preserving quotients, auxiliary structures that will be used in the statement of the main theorem are also developed in this section. The central result of this paper is presented in the main theorem contained in \cref{sec:MainResult}. \Cref{sec:Proof} has the entire proof of the main theorem, concluding remarks and future work are given in \cref{sec:Conc}.
\section{Maths preliminaries}
\label{sec:MathsPrelim}
This section introduces all the relevant constructions from geometric control theory and differential geometry that will be used in the sequel. Our presentation of bundle theory and Ehresmann connections follows the notation and presentation of \cite{Saunders}, the subsection on the geometric theory of partial differential equations summarizes the results of Goldschmidt in \cite{GoldLinPDE},\cite{GoldNonLinPDE}.

\subsection{Differential geometry}
We will assume that all the objects are $C^{\infty}$ unless otherwise stated. Let $M$ be a $m$-dimensional manifold for $x \in M$ the tangent space $T_xM$ is an $m$-dimensional vector space, if in some coordinate chart $x = (x^1, \cdots , x^m)$ then $T_xM = \text{span}\lbrace \frac{\partial }{\partial x^1} , \cdots , \frac{\partial }{\partial x^m} \rbrace$. The disjoint union of all tangent spaces at all points on the manifold $M$ is the $2m$ dimensional tangent bundle of $M$ which we denote as $TM$. To each $x \in M$ we can associate a dual space to the tangent space which is called the co-tangent space denoted $T^{\ast}_xM = \text{span}\lbrace dx^1, \cdots , dx^m \rbrace$, the co-tangent bundle $T^{\ast}M$ can then be defined as the disjoint union of all co-tangent spaces at all points in $M$. The map $\phi : M \mapsto N$ between manifolds induces a map, the push-forward map $\phi_{\ast} : T_xM \mapsto T_{\phi(x)}N$ that maps vectors in $T_xM$ to vectors in $T_{\phi(x)}N$ by action of the Jacobian of the map $\phi$. Conversely the map $\phi$ induces the pull-back map, $\phi^{\ast} : T_{\phi(x)}^{\ast}N \mapsto T_x^{\ast}M$
\begin{equation}
\phi^{\ast}(\omega)(X) = \omega(\phi_{\ast}(X)), \omega \in T^{\ast}_{\phi(x)}N, X \in T_xM.
\end{equation}

Tensors are multilinear maps that are defined on finite copies of the tangent and co-tangent space.

\begin{Defn}\cite{Talpaert}
Let $M$ be a smooth $m$-dimensional manifold, a \textbf{tensor of type $(r, s)$} at $p \in M$ is a real valued $(r + s)$-multilinear map defined on the Cartesian product of $r$ copies of $T^{\ast}_pM$ and $s$ copies of $T_pM$. The set of all $(r,s)$ tensors at $p$ is denoted $T^r_s(T_pM)$. For some $t \in T^r_s(T_pM)$, 
\begin{equation}
t : \underbrace{T^{\ast}_pM \times \cdots \times T^{\ast}_pM}_{r-copies}\times \underbrace{T_pM \times \cdots \times T_pM}_{s-copies} \mapsto \mathbb{R}
\end{equation}
\end{Defn}

It is possible to multiply tensors of different types via the tensor product operator which is denoted by the symbol $\otimes$. Let $t_1$ be a type $(r, s)$ tensor and $t_2$ be a type $(p, q)$ tensor, the tensor product of $t_1$ and $t_2$ denoted $t_1\otimes t_2$ is a $(r+p, s+q)$ tensor defined as follows,
\begin{multline}
(t_1\otimes t_2)(\omega_1, \cdots, \omega_r, \omega_{r+1}, \cdots, \omega_{r+p},X_1, \cdots, X_s, X_{s+1}, \cdots, X_{s+q}) \\
= t_1(\omega_1, \cdots, \omega_r,X_1, \cdots, X_s)t_2(\omega_{r+1}, \cdots, \omega_{r+p},X_{s+1}, \cdots, X_{s+q})
\end{multline}
where $\omega_i \in T_p^{\ast}M$ and $X_j \in T_pM$ for $i = \{1, \cdots r+p\}$ and $j = \{1, \cdots, s+q\}$.

Fibre bundles generalize the familiar notion of product spaces, formally a fibre bundle is defined as follows.

\begin{Defn}\cite{Steenrod}
A \textbf{fibre bundle} is a 4-tuple $(B, M, \pi, F)$ where
\begin{enumerate}
\item $B$, $M$ and $F$ are smooth manifolds called the \textbf{total space}, \textbf{base space} and \textbf{typical fibre space} respectively,
\item $\pi : B \mapsto M$ is a surjective map called the \textbf{projection},
\item Let $\{V_j\}$ be a family of open sets covering $M$ with $j \in J \subset \mathbb{N}$. For each $j \in J$ there exists a homeomorphism $\phi_j : V_j\times F \mapsto \pi^{-1}(V_j)$.
\end{enumerate}
\end{Defn}

For brevity the projection $\pi$ will be used to identify the fibre bundle $(B, M, \pi, F)$. Consider the general fibre bundle $(B, M, \pi, F)$ let dim$(M)$ = $m$, dim$(B)$ = $m+n$ and $(U, \psi)$ be a coordinate chart of $B$ such that $\psi : U \subset B \mapsto \mathbb{R}^{m+n}$. $(U,\psi)$ is called an \textbf{adapted coordinate chart} if for $p, p' \in U$ and $\pi(p) = \pi(p')$ then $pr_1(\psi(p)) = pr_1(\psi(p'))$, where $pr_1$ is the projection to $\mathbb{R}^m$ \cite{Saunders}.

\begin{Defn}\cite{Saunders}
A map $\phi : M \mapsto B$ is called a \textbf{section} of $\pi$ if $\pi \circ \phi = id_M$. The set of all smooth sections will be denoted $\Gamma(\pi)$. 
\end{Defn}
Bundle morphisms are maps between fibre bundles that preserve the fibre bundle structure. Preserving the fibre bundle structure means that for any two points of the total space that lie on the same fibre, their image must also lie on the same fibre. 

\begin{Defn}\cite{Saunders}
If $(B, M, \pi, F)$ and $(E, N, \rho, H)$ are fibre bundles then a \textbf{bundle morphism} is a pair of maps $(f, \bar{f})$ where $f:B \mapsto E$, $\bar{f}:M \mapsto N$ and $\rho \circ f = \bar{f} \circ \pi$.
 
\begin{center}
\begin{tikzcd}
B \arrow[rrr, "{f}"] \arrow[dd, "{\pi}"] & & & E \arrow[dd, "{\rho}"]  \\
& & & \\
M \arrow[rrr,"\bar{f}"] & & & N
\end{tikzcd}
\end{center}

\end{Defn}
 An example of a bundle morphism is the tangent map $(Tf, f) : (TM, M, \tau_M, \mathbb{R}^m) \mapsto (TN, N, \tau_N, \mathbb{R}^n)$.

Consider the fibre bundle $(B, M, \pi, F)$ let $(\mathcal{U}, x^i)$ be some chart of $M$ which induces an adapted coordinate chart $(\pi^{-1}(\mathcal{U}), x^i, u^{\alpha})$ for $B$. For some $q \in \pi^{-1}(\mathcal{U})$ there is a canonical tangent sub-space $V_q\pi \subset T_qB = \text{ker}T_q\pi$ which will be referred to as the \textbf{vertical sub-space}. The disjointed union of these vertical sub-spaces defines the vertical sub-bundle of $\pi$ denoted $V\pi$. An Ehresmann connection represents a non-canonical way to specify a sub-space complementary to the canonical vertical sub-space.

\begin{Defn}\cite{Saunders}
A \textbf{connection} on the fibre bundle $(B, M, \pi, F)$ can be equivalently defined as 
\begin{enumerate}
\item a smooth distribution $H\pi \subset TB$ called the horizontal sub-bundle such that $TB = V\pi \oplus H\pi$,
\item a smooth vector bundle homomorphism $K : TB \mapsto TB$ for which $K(TB) = V\pi$ and $K\circ K = K$.
\end{enumerate}
\end{Defn}

The connection allows one to relate vectors on the base manifold $M$ to vectors on the total space $B$ in the following way. Consider $q \in B$ and $p \in M$ such that $\pi(q) = p$, the restriction of the tangent map $T\pi$ to the horizontal sub-space $H_q\pi$ is an isomorphism $T\pi : H_q\pi \mapsto T_pM$. The inverse to this isomorphism is called the horizontal lift $\text{Hor}_q : T_{\pi(q)}M \mapsto H_q\pi$, this map is uniquely defined by the connection and provides a way of ``\textit{lifting}'' vectors from $TM$ to $TB$.

Consider the adapted coordinate chart $(\mathcal{W}, (x^i, u^{\alpha}))$ with $q \in \mathcal{W}$. In these coordinates the vertical sub-space takes the simple form $V\pi = \text{span}\{\frac{\partial}{\partial u^{\alpha}}\}$. The vector bundle homomorphism $K$ can be viewed as a $V\pi$ valued one-form on $B$, in the adapted coordinate system this gives

\begin{equation}
K = \left( du^{\alpha} -\Gamma^{\alpha}_i(x^i, u^{\alpha})dx^i \right) \otimes\frac{\partial }{\partial u^{\alpha}} .
\end{equation}

The functions $\Gamma^{\alpha}_i(x^i, u^{\alpha})$ uniquely define the connection, in these coordinates the horizontal sub-bundle takes the following form,

\begin{equation}
H\pi = \text{span}\{ \frac{\partial}{\partial x^i} + \Gamma^{\alpha}_i(x^i, u^{\alpha})\frac{\partial}{\partial u^{\alpha}} \}.
\end{equation}
 
The horizontal lift map can also be viewed as a $H\pi$ valued one-form on $M$ which is written as

\begin{equation}
\text{Hor}_q = dx^i\otimes\left(\frac{\partial }{\partial x^i} + \Gamma^{\alpha}_i(x^i, u^{\alpha})\frac{\partial }{\partial u^{\alpha}} \right).
\end{equation}

 This can also be written conveniently in matrix form as shown below.

\begin{equation}
\text{Hor}_q = 
\bbordermatrix{~ & 1 &        & m \cr
							 1 & 1 & \cdots & 0 \cr
							 \vdots & \vdots & \ddots & \vdots \cr
							 m & 0 & \cdots & 1 \cr
							 1 & \Gamma^1_1 & \cdots & \Gamma^1_m \cr
							 \vdots & \vdots & \ddots & \vdots \cr
							 n & \Gamma^{n}_1 & \cdots & \Gamma^{n}_{m} \cr}
\end{equation}

Where $\text{dim}(M) = m, \text{dim}(B) = m+n$. An important property of a connection is its curvature defined below.

\begin{Defn}\cite{Saunders}
The curvature tensor of the connection is the $(1,2)$-tensor $R : \mathcal{X}(B)\times \mathcal{X}(B) \mapsto \mathcal{X}(B)$ defined by 

\begin{equation}
R(X, Y) = K \left( \left[ (X - K(X)), (Y - K(Y)) \right] \right), \quad X, Y \in \mathcal{X}(B).
\end{equation}
\end{Defn}

The coordinate expression of the curvature tensor $R$ is
\begin{equation}
R = \frac{1}{2}\left(\frac{\partial \Gamma^{\alpha}_{i_2}}{\partial x^{i_1}} + \Gamma^{\alpha_1}_{i_1}\frac{\partial \Gamma^{\alpha}_{i_2}}{\partial u^{\alpha_1}} - \frac{\partial \Gamma^{\alpha}_{i_1}}{\partial x^{i_2}} - \Gamma^{\alpha_1}_{i_2}\frac{\partial \Gamma^{\alpha}_{i_1}}{\partial u^{\alpha_1}} \right) dx^{i_1} \wedge dx^{i_2}\otimes \frac{\partial }{\partial u^{\alpha}}.
\end{equation}

\subsection{Jet bundles}
\label{subsec:jetbundles}
The jet bundle formalism provides a geometric way of describing partial differential equations and will play a central role in the main result of this work. This presentation of the theory of jet bundles follows closely the approach of D.J.Saunders in \cite{Saunders}, the interested reader can consult this source for a more detailed coverage of these ideas.

Let $(B, M, \pi, F)$ be a fibre bundle, $p \in M$. The local sections $\phi, \psi \in \Gamma_p(\pi)$ are said to be locally $k$-equivalent at the point $p$ if $\phi(p) = \psi(p)$ and if their derivatives up to the $k^{th}$ order are equal. If $(x^i, u^{\alpha})$ is some adapted coordinate system in some neighbourhood of $\phi(p)$ then $\phi$ and $\psi$ are said to $k^{th}$-equivalent if

\begin{equation}
\frac{\partial^j \phi^{\alpha} }{\partial x^{i_1} \cdots \partial x^{i_j} } = \frac{\partial^j \psi^{\alpha} }{\partial x^{i_1} \cdots \partial x^{i_j} }.
\end{equation}  

Where $j \in (1, \cdots , k)$, $i_1\le \cdots \le i_j$, $i_1, \cdots , i_j \in (1, \cdots , dim(M)=m)$, $\alpha \in (1, \cdots , dim(F) = n)$. The $k$-equivalence set at $p$ containing $\phi$ is called the $k$-\textbf{jet} of $\phi$ and is denoted $j^k_p\phi$.

\begin{Defn}\cite{Saunders}
The \textbf{$k^{th}$-jet manifold} of $(B, M, \pi, F)$ is the set of all $k$-jets and is denoted $J^k\pi$ 
\begin{equation}
J^k\pi = \bigcup _{p \in M}\lbrace j^k_p\phi \mid \phi \in \Gamma_p(\pi) \rbrace.
\end{equation}
\end{Defn}

 The $k^{th}$-jet bundle is equipped with maps $\pi_k$ and $\pi_{k,0}$ called the \textbf{source} and \textbf{target projections} respectively, these maps are defined as follows
\begin{eqnarray*}
\pi_k &:& J^k\pi \mapsto M\\
& & j^k_p\phi \mapsto p
\end{eqnarray*}
and
\begin{eqnarray*}
\pi_{k,0} &:& J^k\pi \mapsto B\\
& & j^k_p\phi \mapsto \phi(p).
\end{eqnarray*}
If the bundle $\pi$ has the adapted coordinates $(x^i, u^{\alpha})$ is some open set $W \subset B$, then the $k^{th}$ jet bundle $J^k\pi$ has the induced coordinates $(x^i, u^{\alpha}, u^{\alpha}_j)$ where $j = (1, \cdots, k)$. Consider $j^k_p\phi \in J^k\pi$ in these induced coordinates we have, $x^i(j^k_p\phi) = x^i(p), u^{\alpha}(j^k_p\phi) = u^{\alpha}(\phi(p))$ and
\begin{equation}
u^{\alpha}_j(j^k_p\phi) = \frac{\partial^j \phi^{\alpha}}{\partial x^{i_1} \cdots \partial x^{i_j}} . \nonumber
\end{equation}

$J^k\pi$ is a manifold in its own right, additionally $J^k\pi$ can be equipped with a fibre bundle structure. Thus $J^k\pi$ can be viewed as the total space over the base manifolds $J^l\pi$, $B$ and $M$. $\pi_{k, k-1} : J^k\pi \mapsto J^{k-1}\pi$ is not just a bare fibre bundle but it is actually an affine bundle. The bundle $\pi_{k, k-1}$ is an affine bundle modelled on the vector bundle $\pi_{k-1}^{\ast}(S^kT^{\ast}M) \otimes \pi_{k-1, 0}^{\ast}(V\pi)$ where $S^kTM$ is the $k$-symmetric tensor bundle and $V\pi$ is the vertical bundle.
\subsection{Geometric partial differential equations}
\label{subsec:GeomPDE}
This section presents the geometric theory of partial differential equations as developed by Goldschmidt\cite{GoldLinPDE},\cite{GoldNonLinPDE}. The results presented here play a pivotal role in the main result of this paper, for a more in-depth coverage of the material consulting the papers of Goldschmidt \cite{GoldLinPDE}\cite{GoldNonLinPDE} is highly recommended. 

\begin{Defn}\cite{GoldNonLinPDE}
Let $(B, M, \pi, F)$ and $(E, M, \tilde{\pi}, G)$ be fibre bundles. A \textbf{partial differential equation} of order $k$ is a fibred embedded sub-manifold $R_k \subset J^k\pi$. Additionally there always exists a fibre bundle morphism $\Phi : J^k\pi \mapsto \tilde{\pi}$ such that $R_k = \text{ker}\:\Phi$.
\end{Defn}
The solution of a partial differential equation is then defined as follows.
\begin{Defn}\cite{Saunders}
Let $(B, M, \pi, F)$ be a fibre bundle and let $R_k \subset J^k\pi$ be a $k^{th}$-order partial differential equation. A \textbf{solution} of $R_k$ is a local section $\phi : \mathcal{U} \subset M \mapsto B$ such that $j_p^k\phi \in R_k$ for every $p \in \mathcal{U}$.
\end{Defn}

A differential equation can be differentiated to produce another higher order differential equation. This action is referred as prolonging the differential equation. Within the geometric setting prolongation is defined as follows.

\begin{Defn}\cite{Saunders}
Let $(B, M, \pi, F)$ and $(E, M, \tilde{\pi}, G)$ be fibre bundles and consider the $k^{th}$-order differential equation $R_k \subset J^k\pi$ defined by the fibre bundle morphism $\Phi : J^k\pi \mapsto E$. The $l^{th}$ prolongation of $R_k$ is the $(k+l)^{th}$-order differential equation $R_{k+l} \subset J^{k+l}\pi$ defined by the fibre bundle morphism $\rho_l(\Phi) : J^{k+l}\pi \mapsto J^{l}\tilde{\pi}, \rho_l(\Phi)(j^{k+l}_p\phi) = j^l_p(\Phi(j^k_p\phi))$. Where $p \in M$ and $\phi$ is a section of $\pi$.
\end{Defn}

The symbol of the differential equation is a structure that encodes information about the highest order elements in the linearization of the differential equation\cite{GoldNonLinPDE}.  Before stating the definition recall that $\pi_{k,k-1} : J^k\pi \mapsto J^{k-1}\pi$ is an affine fibre bundle that is modelled on the vector bundle over $J^{k-1}\pi$ with total space $\pi^{\ast}_{k-1}\left(S^kT^{\ast}M\right) \otimes \pi^{\ast}_{k-1,0}\left(V\pi\right)$. There exists a canonical inclusion map

\begin{equation}
\epsilon_k : \pi^{\ast}_{k} \left( S^kT^{\ast}M \right) \otimes \pi^{\ast}_{k,0}\left( V\pi \right) \mapsto V\pi_k.
\end{equation}
Assuming $(x^i, u^{\alpha})$ are adapted coordinates for some chart of $\pi$ the coordinate expression for $\epsilon_k$ is

\begin{equation}
\epsilon : \xi^{\alpha}_{i_1\cdots\i_k}dx^{i_1}\otimes \cdots \otimes dx^{i_k}\otimes \frac{\partial }{\partial u^{\alpha}} \mapsto \xi^{\alpha}_{i_1\cdots\i_k} \frac{\partial }{\partial u^{\alpha}_{i_1\cdots\i_k}}.
\end{equation}

\begin{Defn}\cite{GoldNonLinPDE}
Let $(B, M, \pi, F)$ and $(E, M, \tilde{\pi}, G)$ be fibre bundles and consider the $k^{th}$-order differential equation $R_k \subset J^k\pi$ defined by the fibre bundle morphism $\Phi : J^k\pi \mapsto E$. The \textbf{symbol} of $R_k$ denoted $\sigma(\Phi)$ is the vector bundle morphism $\sigma(\Phi) = V\Phi\circ\epsilon_k : \pi^{\ast}_{k} \left( S^kT^{\ast}M \right) \otimes \pi^{\ast}_{k,0}\left( V\pi \right) \mapsto V\tilde{\pi}$. Where $V\Phi$ is the restriction of the tangent map $T\Phi$ to $V\pi_k$. Let $\text{ker}\:\rho(\sigma(\Phi)) = G_k$, at times we refer to $G_k$ as the symbol of $R_k$.
\end{Defn}

The prolongation of the symbol of a partial differential equation is defined as follows.
\begin{Defn}\cite{BahmanThesis}
Let $(B, M, \pi, F)$ and $(E, M, \tilde{\pi}, G)$ be fibre bundles and consider the $k^{th}$-order differential equation $R_k \subset J^k\pi$ defined by the fibre bundle morphism $\Phi : J^k\pi \mapsto E$. For $p_k \in R_k$ the $l^{th}$-prolongation of the symbol $\sigma(\Phi)|_{p_k}$ is the map

\begin{equation}
\rho_l(\sigma(\Phi)|_{p_k}) : S^{k+l}T^{\ast}_{\pi_k(p_k)}M\otimes V_{\pi_{k,0}(p_k)}\pi \mapsto S^lT^{\ast}_{\pi_k(p_k)}M\otimes V_{\Phi(p_k)}\tilde{\pi},
\end{equation}

defined by $(\text{id}_{S^lT^{\ast}_{\pi_k(p_k)}M}\otimes \sigma(\Phi)|_{p_k}) \circ (\Delta_{k,l}\otimes \text{id}_{V\pi})$. Where $\Delta_{k,l} : S^{k+l}T^{\ast}_{\pi_k(p_k)}M \mapsto S^lT^{\ast}_{\pi_k(p_k)}M \otimes S^kT^{\ast}_{\pi_k(p_k)}M$ is the natural inclusion. Let $\text{ker}(\rho_l(\sigma(\Phi)|_{p_k}) = G_{k+l}$, at times we refer to $G_{k+l}$ as the $l^{th}$ prolongation of the symbol of $R_k$.
\end{Defn}

The notion of a formal solution formalizes the idea of approximating the solution of a partial differential equation by a finite order Taylor series.

\begin{Defn}\cite{Gasqui}
Let $(B, M, \pi, F)$ and $(E, M, \tilde{\pi}, G)$ be fibre bundles and consider the $k^{th}$ order differential equation $R_k \subset J^k\pi$ defined by the fibre bundle morphism $\Phi : J^k\pi \mapsto E$. A \textbf{local formal solution} of order k is a local section of $R_k$ i.e  $\phi_k \in \Gamma(\pi_k|_\mathcal{U})$, $\phi_k : \mathcal{U} \subset M \mapsto R_k \subset J^k\pi$.
\end{Defn} 

The process of constructing the Taylor series solution of a differential equation can only be successful if a formal solution of order $k$ can be prolonged to a formal solution of higher order. This quality of being able to iteratively construct Taylor series solutions is the essence of the concept of formal integrability defined below.

\begin{Defn}\cite{GoldLinPDE}
Let $(B, M, \pi, F)$ and $(E, M, \tilde{\pi}, G)$ be fibre bundles and consider the $k^{th}$-order differential equation $R_k \subset J^k\pi$ defined by the fibre bundle morphism $\Phi : J^k\pi \mapsto E$. The partial differential equation $R_k$ is \textbf{formally integrable} if $R_{k+l}$ is a fibred submanifold and if the maps $\pi_{k+l, k} : R_{k+l} \mapsto R_k$ are epimorphisms for $l \in \mathbb{Z}_{>0}$.
\end{Defn} 

From the above definition the property of being formally integrable is un-testable as it involves checking the surjectivity of an infinite number of maps. The central result of the geometric theory of partial differential equations developed in \cite{GoldNonLinPDE} provides testable conditions for formal integrability. Before stating this theorem some prerequisite definitions and theorems are required.

\begin{Defn}\cite{Gasqui}
Let $(B, M, \pi, F)$ and $(E, M, \tilde{\pi}, G)$ be fibre bundles and consider the $k^{th}$-order differential equation $R_k \subset J^k\pi$ defined by the fibre bundle morphism $\Phi : J^k\pi \mapsto E$. Let $G_k$ be the symbol of $R_k$, for $p_k \in R_k$ the basis $\lbrace e^1, \dots, e^m\rbrace$ of $T^{\ast}_{\pi_k(p_k)}M$ is called \textbf{quasi-regular} if

\begin{equation}
\text{dim}(G_{k+1}|_{p_{k+1}}) = \text{dim}(G_k|p_k) + \sum^{m-1}_{j=1}\text{dim}(G_{k,j}|_{\pi_k(p)}).
\end{equation}
 Where $G_{k,j}|_{\pi_k(p)}$ is given by 
\begin{equation}
G_{k,j}|_{\pi_k(p)} = G_k|_{p_k} \cap S^k\Sigma_j|_{\pi_k(p_k)}, \quad \Sigma_j = \text{span}\lbrace e^{j+1}, \dots, e^m \rbrace.
\end{equation}
\end{Defn} 

\begin{theorem}\cite{Guill}
\label{Thrm:InvolutiveThrm}
Let $(B, M, \pi, F)$ and $(E, M, \tilde{\pi}, G)$ be fibre bundles and consider the $k^{th}$ order differential equation $R_k \subset J^k\pi$ defined by the fibre bundle morphism $\Phi : J^k\pi \mapsto E$. If there exists a quasi-regular basis for $T_{\pi_k(p_k)}^{\ast}M$ where $p_k \in    R_k$, then the symbol $G_k$ is said to be involutive.
\end{theorem}

We can now state the central theorem that allows the development of integrability conditions for partial differential equations.

\begin{theorem}\cite{GoldNonLinPDE}
Let $(B, M, \pi, F)$ and $(E, M, \tilde{\pi}, G)$ be fibre bundles and consider the $k^{th}$ order differential equation $R_k \subset J^k\pi$ defined by the fibre bundle morphism $\Phi : J^k\pi \mapsto E$. If 
\begin{enumerate}
\item $R_{k+1}$ is a fibred submanifold of $J^{k+1}\pi$
\item $\pi_{k+1, k} : R_{k+1} \mapsto R_k$ is surjective
\item $G_k$ is involutive
\end{enumerate}
then $R_k$ is formally integrable.
\end{theorem}

This theorem only requires prolonging the partial differential equation once and testing if the prolonged differential equation projects onto the original partial differential equation. Requiring $\pi_{k+1, k}$ to be surjective can be shown to be equivalent to requiring the zeroing of the so-called curvature map $\kappa : R_k \subset J^k\pi \mapsto S^lT^{\ast}M\otimes V\tilde{\pi}/\text{Im}\rho_1(\sigma(\Phi))$ \cite{GoldNonLinPDE}\cite{Gasqui}. Let $p_k \in R_k$ and $p_{k+1} \in J^{k+1}\pi$ be such that $\pi_{k+1, k}(p_{k+1}) = p_k$ the curvature map is defined as

\begin{eqnarray}
\kappa(p_k) &=& \tau\left(\rho_1(\Phi)(p_{k+1}) - j^1\Phi(p)\right),\\
\tau &:& S^lT^{\ast}M\otimes V\tilde{\pi} \mapsto S^lT^{\ast}M\otimes V\tilde{\pi}/\text{Im}\rho_1(\sigma(\Phi)).
\end{eqnarray}

Where $\tau$ is the canonical projection map onto $S^lT^{\ast}M\otimes V\tilde{\pi}/\text{Im}\rho_1(\sigma(\Phi))$.
\subsection{Control theory}
\label{subsec:ControlTheory}
\begin{Defn}
A \textbf{control system} is a 5-tuple $\Sigma = (B, M, \pi_M, U, F)$ where the 4-tuple $(B, M, \pi_M, U)$ is a fibre bundle and a smooth map $F : B \mapsto TM$ such that $\pi_{TM}\circ F = \pi_M$ where $\pi_M$ and $\pi_{TM}$ are the canonical projections of the fibre bundle and the tangent bundle respectively. 
\end{Defn} 

The base manifold $M$ models the state space of the control system and the typical fibre models the control input space. Locally the total space looks like a product space of the state and control input space however globally the topology can change drastically, allowing the model to accommodate instances where the control input space depends on the state space in a non-trivial way. By choosing fibre respecting coordinates for $B$ the usual representation of the control system as a set of differential equations can be easily recovered.Within this framework of control system representation trajectories of a control system are given by the following definition\cite{Tabauda}.     

\begin{Defn}
Let $\Sigma = (B, M, \pi_M, U, F)$ be a control system, the smooth curve $\gamma^M(t) : \mathbb{R} \mapsto M$ is called a trajectory of $\Sigma$ if there exists a curve $\gamma^{B}(t) : \mathbb{R} \mapsto B$ such that:
\begin{enumerate}
\item $\pi_M \circ \gamma^B(t) = \gamma^M(t)$
\item $\frac{d}{dt}\gamma^M(t) = F\circ\gamma^B(t)$
\end{enumerate}
\end{Defn}

Consider two control systems  $\Sigma = (B, M, \pi_M, U, F)$ and $\tilde{\Sigma} = (\tilde{B}, N, \pi_N, V, G)$, let $\Sigma_{\text{traj}}$ and $\tilde{\Sigma}_{\text{traj}}$ be the set of trajectories of the control systems respectively. $\Sigma$ and $\tilde{\Sigma}$ are said to be equivalent if and only if the sets $\Sigma_{\text{traj}}$ and $\tilde{\Sigma}_{\text{traj}}$ can be put in one-to-one correspondence\cite{ElkinDecomp}.

\begin{proposition}
The control systems $\Sigma$ and $\tilde{\Sigma}$ are equivalent if there exists a bundle isomorphism $\Phi = (\phi, \psi)$ such that the following diagram commutes.

\begin{center}
\begin{tikzcd}
B \arrow[rrr, "\psi"] \arrow[dd, "{\pi_M}"] \ar[dr, "F"]
& & & \tilde{B} \arrow[dd, "{\pi_N}"] \ar[dl, "{G}"] \\
 & TM \arrow[r, "{T\phi}"] \ar[dl, "{\pi_{TM}}"] & TN \ar[dr, "{\pi_{TN}}"] & \\
M \arrow[rrr,"\phi"] & & & N
\end{tikzcd}
\end{center}
 
\end{proposition}

\begin{proof}
Let $\gamma^M(t)$ be a trajectory of $\Sigma$ by definition there exists a curve $\gamma^B(t)$ such that $\pi_M \circ \gamma^B(t) = \gamma^M(t)$. From the fibre preserving property,
\begin{eqnarray*}
\pi_N\circ \psi \circ \gamma^B(t) &=& \phi \circ \pi_M \circ \gamma^B(t)\\
\pi_N\circ( \psi \circ \gamma^B)(t) &=& \phi \circ \gamma^M(t) 
\end{eqnarray*}
This proves the first part of the trajectory definition. For the second part differentiate the curve $\phi \circ \gamma^M(t)$.
\begin{eqnarray*}
\frac{d}{dt}(\phi \circ \gamma^M(t)) &=& T\phi\circ \frac{d}{dt}(\gamma^M(t))\\
  &=& T\phi \circ F \circ \gamma^B(t)\\
	&=& G \circ \psi \circ \gamma^B(t)\\
	&=& G \circ (\psi \circ \gamma^B)(t)
\end{eqnarray*}

Therefore $\phi\circ \gamma^M(t)$ is a trajectory of $\tilde{\Sigma}$. Since $\Phi$ is a bundle isomorphism there exists a smooth inverse bundle morphism $\Phi^{-1} : \pi_N \mapsto \pi_M$. By following the exact same steps as above it can be shown that if $\gamma^N(t)$ is a trajectory of $\tilde{\Sigma}$ then $\phi^{-1}(\gamma^N(t))$ is a trajectory of $\Sigma$.
\end{proof}

The notion of control system quotients formalises the idea of abstracting/reducing a control system, the quotient control system is a lower order approximation of the original control system where some of the information in the original system has been factored out.

\begin{Defn}\cite{Tabauda}
Consider the control systems $\Sigma_M = (B, M, \pi_M, U, F)$ and $\tilde{\Sigma} = (\tilde{B}, N, \pi_N, V, G)$ where $\text{dim}(M) > \text{dim}(N)$, $\tilde{\Sigma}$ is a \textbf{quotient control system} of $\Sigma$ if there exists a fibre bundle morphism $\Phi = (\phi, \psi) : \pi_M \mapsto \pi_N$ which satisfies the following conditions,

\begin{enumerate}
\item The maps $\phi : M \mapsto N$ and $\psi : B \mapsto \tilde{B}$ are surjective submersions i.e $\Phi$ is a bundle epimorphism.
\item $(\phi, \psi)$ maps trajectories of $\Sigma$ to trajectories of $\tilde{\Sigma}$.
\end{enumerate}

\end{Defn}

An interesting property of quotients of control systems proved in Tabauda et.al\cite{Tabauda} is the surprising fact that for any control system $\Sigma$ existence of a quotient control system is guaranteed under really mild conditions. The theorem is stated here with out proof.
\begin{theorem}\cite{Tabauda}
Consider the control system $\Sigma = (B, M, \pi_M, U, F)$ and $\phi : M \mapsto N$ a surjective submersion, if $T\phi\circ F : B \mapsto TN$ has constant rank and connected fibres then there exists,
\begin{enumerate}
\item a control system $\tilde{\Sigma} = (\tilde{B}, N, \pi_N, V, G)$,
\item a fibre preserving lift $\psi : B \mapsto \tilde{B}$ of $\phi$ such that $\tilde{\Sigma}$ is a quotient control system of $\Sigma$ with fibre bundle morphism $(\phi, \psi)$. 
\end{enumerate}
\end{theorem}
If the control system $\Sigma$ has a constant rank control distribution and a none vanishing drift vector field (i.e $F$ is a constant rank map) then the above theorem guarantees the existence of control system on the manifold $N$ which is a quotient of $\Sigma$. 

With the requisite mathematical constructions having been presented it becomes easier to see that most of the reduction techniques employed in control theory are actually instances of quotients. Symmetry based reduction can be viewed as quotienting where the map $\phi$ is the projection to the quotient manifold generated by the Lie group action on the state space manifold\cite{GrizzleSymm}. The same goes for system decomposition via controlled invariance distributions, the state-space quotienting map $\phi$ corresponds to the projection map from the state space manifold to its quotient sub-manifold generated by factoring out the integral sub-manifolds of the controlled invariant distribution\cite{IsidoriDecoup}.  

\section{Problem statement}
\label{sec:ProbStat}
Using the language of geometric control theory and the characterization of quotients developed in the above sections the question of stabilizability preserving quotients can be expressed formally as follows.
\begin{Prob}
\label{Problem:GenCase}
Consider two control affine systems $\Sigma = (B, M, \pi_M, U, F)$ and $\tilde{\Sigma} = (\tilde{B}, N, \pi_N, V, G)$. Assume that 
\begin{enumerate}
\item dim(M) $>$ dim(N),
\item the control distribution defined by $F\circ \pi_M^{-1}(p) \subset T_pM$ is constant rank and smooth,
\item $\tilde{\Sigma}$ is a quotient control system of $\Sigma$ under the action of the smooth fibre bundle morphism $\Phi = (\phi, \psi)$
\item the fibre bundle morphism $\Phi$ in adapted coordinates $(x^i, u^j)$ and $(y^q, v^k)$ for $\pi_M$ and $\pi_N$ respectively has the following form
\begin{eqnarray}
\phi &:& (x^i) \mapsto (y^q = x^q)\\
\psi &:& (x^i, u^j) \mapsto (y^q = x^q, v^k = \varphi^k(x^i) + u^j\beta^k_j(x^i))
\end{eqnarray}
where $i = 1, \cdots , m$, $j = 1, \cdots , r$, $q = 1, \cdots , n$, $k = 1, \cdots , s$ and the functions $\varphi^k(x^i)$, $\beta^k_j(x^i)$ are all smooth.
\item  $\tilde{\Sigma}$ is stabilizable i.e there exists a smooth section $\alpha$ of $\pi_N$ and a positive definite function $\tilde{V} : N \mapsto \mathbb{R}$ such that the closed loop dynamics $G\circ \alpha(q), q \in N$ are asymptotically stable in the Lyapunov sense about the equilibrium point $q_0$. 
\end{enumerate}
\textbf{Under what conditions is $\Sigma$ locally stabilizable about the equilibrium point $p_0 \in \phi^{-1}(q_0)$.}
\end{Prob}

Let $\Sigma = (B, M, \pi_M, U, F)$ and $\tilde{\Sigma} = (\tilde{B}, N, \pi_N, V, G)$ be control affine systems. Assume there are coordinate charts $\mathcal{O} \subset B$, $\mathcal{Q} \subset \tilde{B}$ such that there are coordinates $\mathbf{x} = (x^1, \cdots , x^m, u^1, \cdots , u^r)$ and $\mathbf{y} = (y^1, \cdots , y^n, v^1,\cdots , v^s)$ for $\mathcal{O} \subset B$ and $\mathcal{Q} \subset \tilde{B}$ respectively such that $q_0 = \mathbf{0}$ and $p_0 = \mathbf{0}$. The control systems $\Sigma$ and $\tilde{\Sigma}$ can be written in the following form, 

\begin{eqnarray}
\label{Eqn:SigmaEqns}
\Sigma &:& 
\begin{bmatrix}
\dot{x}^1\\
\vdots\\
\dot{x}^m
\end{bmatrix} = 
\begin{bmatrix}
f^1_0(\mathbf{x})\\
\vdots\\
f^m_0(\mathbf{x})
\end{bmatrix} +
\begin{bmatrix}
f^1_1(\mathbf{x})\\
\vdots\\
f^m_1(\mathbf{x})
\end{bmatrix}u^1
+ \cdots +
\begin{bmatrix}
f^1_r(\mathbf{x})\\
\vdots\\
f^m_r(\mathbf{x})
\end{bmatrix}u^r. \\
\tilde{\Sigma} &:&
\begin{bmatrix}
\dot{y}^1\\
\vdots \\
\dot{y}^n
\end{bmatrix} = 
\begin{bmatrix}
g^1_0(\mathbf{y})\\
\vdots \\
g^n_0(\mathbf{y})
\end{bmatrix} +
\begin{bmatrix}
g^1_1(\mathbf{y})\\
\vdots \\
g^n_1(\mathbf{y})
\end{bmatrix}v^1 +
\cdots +
\begin{bmatrix}
g^1_s(\mathbf{y})\\
\vdots \\
g^n_s(\mathbf{y})
\end{bmatrix}v^s \label{Eqn:TildeSigmaEQN}
\end{eqnarray}

As stated in the problem statement (\ref{Problem:GenCase}) it is assumed that $\tilde{\Sigma}$ is stabilizable about the origin i.e there exists feedback  controls $v^1 = \alpha^1(\mathbf{y}), \cdots , v^s = \alpha^s(\mathbf{y})$ and a control Lyapunov function $\tilde{V} : N \mapsto \mathbb{R}$ such that
\begin{eqnarray}
W(\mathbf{y}) &=& \frac{\partial \tilde{V}}{\partial y^1}\left[ g^1_0(\mathbf{y}) + g^1_1(\mathbf{y})\alpha^1(\mathbf{y}) + \cdots + g^1_s(\mathbf{y})\alpha^s(\mathbf{y}) \right] + \nonumber \\
&\cdots& + \frac{\partial \tilde{V}}{\partial y^n}\left[g^n_0(\mathbf{y}) + g^1_n(\mathbf{y})\alpha^1(\mathbf{y}) + \cdots + g^n_s(\mathbf{y})\alpha^s(\mathbf{y}) \right], 
\end{eqnarray}
where $W(\mathbf{y})$ is a negative definite function. From the control system $\Sigma$'s control vector fields construct the constant rank smooth distribution $C$ defined by,
\begin{equation}
C = \text{span}\lbrace f^i_1(\mathbf{x})\frac{\partial }{\partial x^i}, \cdots , f^i_r(\mathbf{x})\frac{\partial }{\partial x^i} \rbrace.
\end{equation}
Since $C$ is constant rank and smooth it is possible to construct a complementary constant rank and smooth distribution $D$ such that $TM = C \oplus D$, the canonical projection onto the distribution $D$ will be denoted $P_D : TM \mapsto D$. The projection map $P_D$ can be represented as a vector valued one-form, let $D = \text{span}\lbrace e_1, \cdots , e_{m-r} \rbrace$ then $P_D = P_{D, i}^a(\mathbf{x})dx^i\otimes e_a$ for $a = 1, \cdots , m-r$. In matrix form $P_D$ is represented as a $(m-r)\times m$ matrix where the co-efficient $P^a_{D,i}$ corresponds to the $(a,i)$ matrix element.

The fibre bundle $\phi : M \mapsto N$ can be equipped with a connection defined below.

\begin{proposition}
\label{Prop:CoonectionComp}
The fibre bundle defined by the surjective submersion $\phi : M \mapsto N$ can be equipped with an Ehresmann connection which in coordinates $(x^1,\cdots , x^m)$ and $(y^1, \cdots , y^n)$ for $M$ and $N$ respectively has the following equivalent representations.
\begin{enumerate}
\item The canonical vertical bundle of the fibre bundle denoted $VM$ has the form $VM = \text{span}\lbrace \frac{\partial }{\partial x^{n+1}}, \cdots , \frac{\partial }{\partial x^m} \rbrace$. A connection defined as a complementary subspace to the canonical vertical bundle  will be denoted $HM$ and has the following form
\begin{eqnarray}
HM &=& \text{span}\lbrace \frac{\partial }{\partial x^1} + \Gamma^{n+1}_1(\mathbf{x})\frac{\partial }{\partial x^{n+1}} + \cdots + \Gamma^{m}_1(\mathbf{x})\frac{\partial }{\partial x^m} , \nonumber \\
 &\cdots& , \frac{\partial }{\partial x^n} + \Gamma^{n+1}_n\frac{\partial }{\partial x^{n+1}} + \cdots + \Gamma^m_n(\mathbf{x})\frac{\partial }{\partial x^m} \rbrace.
\end{eqnarray}
 \item As a $VM$-valued one form the connection has the following form,
\begin{eqnarray}
K &=& \left( dx^{n+1} - \Gamma^{n+1}_1(\mathbf{x})dx^1 - \cdots - \Gamma^{n+1}_n(\mathbf{x})dx^n\right) \otimes \frac{\partial }{\partial x^{n+1}} + \cdots \nonumber\\
&+& \left( dx^m - \Gamma^m_1(\mathbf{x})dx^1 - \cdots - \Gamma^m_n(\mathbf{x})dx^n \right) \otimes \frac{\partial }{\partial x^m}.
\end{eqnarray}
\item The connection defines the horizontal lift map as a $HM$-valued one-form  denoted $\text{Hor}_{\mathbf{x}}$,
\begin{eqnarray}
\text{Hor}_{\mathbf{x}} &=& dy^1\otimes \left( \frac{\partial }{\partial x^1} + \Gamma^{n+1}_1(\mathbf{x})\frac{\partial }{\partial x^{n+1}} + \cdots +\Gamma^m_1(\mathbf{x})\frac{\partial }{\partial x^m} \right) + \cdots \nonumber \\
&+& dy^n\otimes \left( \frac{\partial }{\partial x^n} + \Gamma^{n+1}_n(\mathbf{x})\frac{\partial }{\partial x^{n+1}} + \cdots + \Gamma^m_n(\mathbf{x})\frac{\partial }{\partial x^m} \right).
\end{eqnarray}

\end{enumerate}
\end{proposition}

\section{Main result}
\label{sec:MainResult}
Our solution to \cref{Problem:GenCase} is constructive and involves determining a function $V : M \mapsto \mathbb{R}$ such that,
\begin{enumerate}
\item $V^{\ast}(\mathbf{x}) = \phi^{\ast}\tilde{V}(\mathbf{x}) + V(\mathbf{x})$ is a positive definite function on $M$ with $\mathbf{d}V \in \mathbf{ann}(HM), V(\mathbf{0}) = 0$.
\item There exists a stabilizing feedback $u = u(\mathbf{x})$ such that $F\circ u(\mathbf{x}) = \Sigma^{target}(\mathbf{x})$ where $\Sigma^{target}$ is a vector field with locally asymptotically stable dynamics defined by 
\begin{equation}
\Sigma^{target}(\mathbf{x}) = \text{Hor}_{\mathbf{x}} \circ G \circ \alpha \circ \phi(x) - \Delta^{\sharp}\circ \mathbf{d}\phi^{\ast}\tilde{V}(\mathbf{x}) - \Delta^{\sharp}\circ \mathbf{d}V(\mathbf{x}).
\end{equation}
Where $\Delta^{\sharp}$ is a bundle isomorphism $\Delta^{\sharp} : T^{\ast}M \mapsto TM$, $\Delta^{\sharp}(x^i, \omega_idx^i) \mapsto (x^i, \omega_{i_1}\delta^{i_1,i}\frac{\partial }{\partial x^i})$. Note $\delta^{i_1,i}$ is the Kronecker delta symbol.
\end{enumerate}

The target dynamics can be shown to be locally asymptotically stable by verifying the fact that the time derivative of the postulated candidate Lyapunov function $V^{\ast}(\mathbf{x})$ is negative definite along the trajectories of $\Sigma^{target}(\mathbf{x})$. Effectively the above requirements translates into a system of under-determined partial differential equations in $V(\mathbf{x})$ of the form
\begin{equation}
\label{EQN:PDE1}
\frac{\partial V}{\partial x^q} + \Gamma^q_p(\mathbf{x})\frac{\partial V}{\partial x^q} = 0,\: q = 1, \cdots , n.\: p = n+1, \cdots , m.
\end{equation}

\begin{equation}
\label{EQN:PDE2}
P_D\circ (\Sigma^{target}(\mathbf{x}) - f_0(\mathbf{x}) = \mathbf{0}.
\end{equation}
\Cref{EQN:PDE1} is equivalent to requiring $\mathbf{d}V \in \mathbf{ann}(HM)$ and \cref{EQN:PDE2} is equivalent to requiring $F\circ u(\mathbf{x}) = \Sigma^{target}(\mathbf{x})$. The main result of this paper contained in the theorem below gives conditions for the existence of analytic solutions to this system of partial differential equations, these integrability conditions when not met can be viewed as obstructions to our proposed method of constructing a stabilizing feedback for $\Sigma$. 

\begin{theorem}
\label{Thrm:MainTheorem}
Given the assumptions stated in the \cref{Problem:GenCase} the system of partial differential equations described by equations \cref{EQN:PDE1,EQN:PDE2} is integrable if the following conditions are met.
\begin{enumerate}
\item The control systems $\Sigma$ and $\tilde{\Sigma}$ are analytic.
\item The fibre bundle $\phi : M \mapsto N$ is equipped with a flat connection.
\item 
\begin{equation}
A_{i_1}^{a_1,a_2} = \sum_{i = 1}^m \left[P^{a_1}_{D,i}(\mathbf{x})\frac{\partial }{\partial x^i}\left(P^{a_2}_{D,i_1}(\mathbf{x})\right) - P^{a_2}_{D,i}(\mathbf{x})\frac{\partial }{\partial x^i}\left(P^{a_1}_{D,i_1}(\mathbf{x})\right) \right] = 0.
\end{equation}

\item Let $X = \text{Hor}_{\mathbf{x}}\circ G \circ \alpha (\phi(\mathbf{x})) - \mathbf{d}\phi^{\ast}(\tilde{V})(\mathbf{x}) - f_0(\mathbf{x})$, 
\begin{equation}
 B^{a_1, a_2} = \sum_{i = 1}^m \sum_{i_1 = 1}^m \left[ P^{a_2}_{D,i}(\mathbf{x})P^{a_1}_{D,i_1}(\mathbf{x}) - P^{a_1}_{D,i}(\mathbf{x})P^{a_2}_{D,i_1}(\mathbf{x}) \right]\frac{\partial }{\partial x^i}(X^{i_1})(\mathbf{x}) = 0.
\end{equation}

\end{enumerate}
\end{theorem}
\section{Proof}
\label{sec:Proof}
The proof  involves representing the system of partial differential equations described by equations \cref{EQN:PDE1,EQN:PDE2} in a geometric fashion and applying the theory of geometric partial differential equations to determine the integrability conditions of the equations.

\subsection{Setting up the partial differential equation}
\label{subsec:settingPDE}
To transform the partial differential equation requirements as expressed in equations \cref{EQN:PDE1,EQN:PDE2} into the language of geometric differential equations we will need the following geometric objects.
\begin{enumerate}
\item The fibre bundle $(M\times \mathbb{R}, M, \pi, \mathbb{R})$ where $\pi$ is the natural projection. A chart $\mathcal{U} \subset M$ with coordinates $x^i$ for $i = 1, \cdots , m$ induces the adapted coordinates $(x^i, V)$ in the open set $\pi^{-1}(\mathcal{U})$.
\item The first jet bundle $J^1\pi$ has the induced coordinates $(x^i, V, V_i)$.
\item A bundle morphism $\Phi_d : J^1\pi \mapsto T^{\ast}M$, $\Phi_d(x^i, V, V_i) \mapsto (x^i, V_idx^i)$.
\item A bundle isomorphism $\Delta^b : T^{\ast}M \mapsto TM$, $\Delta^b(x^i, \omega_idx^i) \mapsto (x^i, \omega_{i_1}\delta^{i_1,i}\frac{\partial }{\partial x^i})$. $\Delta^b$ is the inverse of $\Delta^{\sharp}$.
\item Recall for a connection equipped fibre bundle $\phi : M \mapsto N$, the tangent bundle splits $TM = VM\oplus HM$ where $VM$ is the canonical vertical bundle and $HM$ is the horizontal bundle defined by the connection. Since $VM$ and $HM$ are regular distributions their annihilators provide a splitting of the cotangent bundle $T^{\ast}M = \mathbf{ann}(VM)\oplus \mathbf{ann}(HM)$. 
\item A projection $P_{VM} : T^{\ast}M \mapsto \mathbf{ann}(VM)$ such that $\text{ker}(P_{VM}) = \mathbf{ann}(HM)$. $P_{VM}$ is a $(1,1)$-tensor and has the coordinate expression $P_{VM}(\mathbf{x}) = P^i_{VM, q}(\mathbf{x}) \\ \frac{\partial }{\partial x^i}\otimes dx^q$, where $\mathbf{ann}(VM) = \text{span}\lbrace dx^1, \cdots , dx^n \rbrace$ and $q = 1, \cdots , n$. In matrix form $P_{VM}$ is a $m\times n$ matrix where the co-efficient $P^i_{VM, q}$ corresponds to the $(i, q)$ entry in the matrix. 
\end{enumerate}

\begin{proposition}
\label{Prop:PDEDefn}
Consider the fibred submanifold $\mathcal{R} \subset J^1\pi$ defined as follows,
\begin{equation}
\mathcal{R} = \lbrace (x^i, V, V_i) | P_D\circ (X(x^i) - \Delta^{\sharp}\circ \Phi_d(x^i, V, V_i)) = P_{VM}\circ \Phi_d(x^i, V, V_i) = \mathbf{0} \rbrace.
\end{equation}
\begin{equation}
X \in \mathcal{X}(M) = \text{Hor}_{x^i} \circ G \circ \alpha\circ \phi(x^i) - \Delta^{\sharp}\circ \mathbf{d}\phi^{\ast}\tilde{V}(x^i) - f_0(x^i).
\end{equation}
$\mathcal{R}$ is the geometric representation of the system of partial differential equations \cref{EQN:PDE1,EQN:PDE2}.
\end{proposition}

For the associated fibre bundle morphism to the differential equation $\mathcal{R}$ consider the fibre bundle $(M\otimes\left( D\oplus \mathbf{ann}(VM) \right), M, \tilde{\pi}, \mathbb{R}^{m+n-r})$. The local adapted coordinates for $\tilde{\pi}$ will be denoted $(x^i, \mathcal{X}^a, \omega_q)$. Let $\Psi : J^1\pi \mapsto \tilde{\pi}$ be a fibre bundle morphism defined below.

\begin{equation}
\label{EQN:PdeCoordEqn}
\Psi(x^i, V, V_i) \mapsto (x^i, P^a_{D,i}\delta^{i,i_1}V_{i_1} - P^p_{D,i}X^i, P^i_{VM,q}V_i).
\end{equation}

The fibred submanifold $\mathcal{R}$ is the zero level set of $\Psi$.

\begin{proposition}
\label{Prop:ProlongEqn}
Consider the jet bundles $J^2\pi$ and $J^1\tilde{\pi}$ with the induced adapted local coordinates $(x^i, V, V_i, V_{[i,i_1]}), V_{[i,i_1]} = V_{[i_1,i]}$ and $(x^i, \mathcal{X}^a, \omega_q, \mathcal{X}^a_i, \omega_{q i})$ respectively. The first prolongation of the differential equation $\mathcal{R}$ is defined as the kernel of the prolonged morphism $\rho_1(\Psi) : J^2\pi \mapsto J^1\tilde{\pi}$, 

\begin{eqnarray}
\rho_1(\Psi)(x^i, V, V_i, V_{[i,i_1]}) &=& (x^i, P^a_{D,i}\delta^{i,i_1}V_{i_1} - P^a_{D,i}X^i, P^i_{VM,q}V_i, P^a_{D,i}\delta^{i,i_1}V_{[i,i_1]} \nonumber \\
&+& \frac{\partial }{\partial x^i}\left(P^a_{D, i_1}\right)\delta^{i,i_1}V_i - \frac{\partial }{\partial x^i}(P^a_{D, i_1})X^{i_1} - P^a_{D,i_1}\frac{\partial }{\partial x^i}(X^{i_1}) \nonumber \\
&,& \frac{\partial }{\partial x^{i_1}}(P^{i_1}_{VM,q})V_i + P^{i_1}_{VM, q}V_{[i, i_1]} ). 
\end{eqnarray}
\end{proposition}
\subsection{Symbol of partial differential equation}
\label{subsec:SymbPDE}
The symbol of the partial differential equation $\mathcal{R}$ is the vector bundle morphism $\sigma(\Psi) : V\Psi \circ \epsilon_1 : \pi^{\ast}_1(T^{\ast}M)\otimes \pi^{\ast}_{1,0}(V\pi) \mapsto V\tilde{\pi}$. Identify $T^{\ast}M\otimes V\pi$ with $T^{\ast}M$ and $V\tilde{\pi} \equiv D\oplus\mathbf{ann}(VM)$ can be identified with $\mathbb{R}^{m-r}\oplus \mathbb{R}^n$ with local coordinates $(\mathcal{X}^a, \omega_q)$. In local coordinates the inclusion map $\epsilon_1 : T^{\ast}M \mapsto V\pi^1_0$ has the form 

\begin{equation}
\epsilon_1(\omega_idx^i) \mapsto (\mathbf{0}\frac{\partial }{\partial V} + \omega_{i_1}\delta^{i_1,i}\frac{\partial }{\partial V_i}).
\end{equation}

The symbol vector bundle morphism $\sigma(\Psi)$ then becomes 

\begin{equation}
\sigma(\Psi)(\omega_idx^i) \mapsto (P^a_{D,i}\delta^{i,i_1}\omega_{i_1}, P^i_{VM,q}V_i).
\end{equation}

Recall that alternatively the symbol of $\mathcal{R}$ can be defined as a subbundle  $ G_1 \subset T^{\ast}M\otimes V\pi$ where $G_1 = \text{ker}\: \sigma(\Psi)$. Therefore $G_1$ is defined as 
\begin{equation}
G_1 = \text{ker}\: \sigma(\Psi) = \Delta^b(C)\cap \mathbf{ann}(HM), 
\end{equation}
To prolong the symbol $\sigma(\Psi)$ consider the following identifications, $S^2T^{\ast}M\otimes V\pi$ is identified with $S^2T^{\ast}M$(can be viewed as the set of symmetric $m\times m$ matrices) and $T^{\ast}M\otimes V\tilde{\pi} = T^{\ast}M\otimes (D\oplus \mathbf{ann}(VM))$ is identified with $(T^{\ast}M\otimes D)\oplus (T^{\ast}M\otimes \mathbf{ann}(VM))$. The inclusion $S^2T^{\ast}M \hookrightarrow T^{\ast}M\otimes T^{\ast}M$ in coordinates is defined as 
\begin{equation*}
\Omega_{[i_1,i_2]}dx^{i_1}\otimes dx^{i_2} \in S^2T^{\ast}M \hookrightarrow (\omega^1_{i_1}dx^{i_1})\otimes(\omega^2_{i_2}dx^{i_2}) \in T^{\ast}M\otimes T^{\ast}M|\Omega_{[i_1, i_2]} = \omega^1_{i_1}\omega^2_{i_2}.
\end{equation*}
The first prolongation of the symbol is the vector bundle morphism $\rho_1(\sigma(\Psi)) = id_{T^{\ast}M}\otimes \sigma(\Psi) : S^2T^{\ast}M \mapsto (T^{\ast}M\otimes D)\oplus (T^{\ast}M\otimes \mathbf{ann}(VM))$ defined as 
\begin{equation}
\rho_1(\sigma(\Psi))(\Omega_{[i_1,i_2]}dx^{i_1}\otimes dx^{i_2}) \mapsto (\Omega_{[i,i_1]}\delta^{i_1,i_2}P^a_{D,i_2}dx^{i}\otimes e_a, \Omega_{[i,i_1]}P^{i_1}_{VM,q}dx^i\otimes dx^q).
\end{equation}
In matrix representation the bundle morphism $\rho_1(\sigma(\Psi))$ has the form, 
\begin{equation}
\rho_1(\sigma(\Psi))(\Omega) \mapsto (P_D \Omega^T, \Omega P_{VM}), \Omega \in S^2T^{\ast}M.
\end{equation}
\begin{proposition}
\label{Prop:CokerDefn}
The kernel and co-kernel of $\rho_1(\sigma(\Psi))$ are the sub-bundles defined as
\begin{eqnarray}
G_{1+1} &=& \text{ker}\: \rho_1(\sigma(\Psi)) = S^2(\Delta^b(C)) \cap S^2 \mathbf{ann}(HM) \\
\text{co-ker}\: \rho_1(\sigma(\Psi)) &=& (T^{\ast}M \otimes D)\oplus (T^{\ast}M\otimes \mathbf{ann}(VM))/ \text{Im}\: \rho_1(\sigma(\Psi)) \nonumber \\
 &\equiv& \wedge^2(\Delta^b(D)) \oplus \wedge^2(\mathbf{ann}(VM)).
\end{eqnarray}
\end{proposition}

\begin{proof}
Consider the splitting of the map $\rho_1(\sigma(\Psi)) = (\rho^a, \rho^b)$ with the component maps defined as, 
\begin{eqnarray}
\rho^a &:& S^2T^{\ast}M \hookrightarrow T^{\ast}M\otimes T^{\ast}M \mapsto T^{\ast}M\otimes D \equiv T^{\ast}M\otimes \Delta^{b}(D)\\
\rho^a(\Omega) &=& (id_{T^{\ast}M}\otimes \Delta^b\circ P_D \circ \Delta^{\sharp})(i(\Omega))\\
\rho^b &:& S^2T^{\ast}M \hookrightarrow T^{\ast}M\otimes T^{\ast}M \mapsto T^{\ast}M\otimes \mathbf{ann}(VM)\\
\rho^b(\Omega) &=& (id_{T^{\ast}M}\otimes P_{VM})(i(\Omega)).
\end{eqnarray}
The kernel and co-kernel of $\rho_1(\sigma(\Psi))$ can be expressed in terms of the kernels and co-kernels of the maps $\rho^a$ and $\rho^b$,
\begin{eqnarray}
\text{ker}\: \rho_1(\sigma(\Psi)) &=& \text{ker}\:\rho^a \cap \text{ker}\:\rho^b. \label{EQN:kerRho_1}\\
\text{co-ker}\: \rho_1(\sigma(\Psi)) &=& \text{co-ker}\: \rho^a \cup \text{co-ker}\: \rho^b \label{EQN:cokerRho_1}.
\end{eqnarray}
To characterise the kernel and co-kernel of $\rho^a$ consider the splitting $T^{\ast}M = \Delta^b(C)\oplus \Delta^b(D)$. This induces the following splitting of the vector bundles $S^2T^{\ast}M$, $T^{\ast}M\otimes T^{\ast}M$ and $T^{\ast}M\otimes \Delta^b(D)$.
\begin{eqnarray}
S^2T^{\ast}M &=& S^2(\Delta^b(C)) \oplus S^2(\Delta^b(D)) \oplus (\Delta^b(C) \otimes \Delta^b(D)).\\
T^{\ast}M\otimes T^{\ast}M &=& (\Delta^b(C)\otimes \Delta^b(C))\oplus (\Delta^b(C)\otimes \Delta^b(D))\oplus (\Delta^b(D)\otimes \Delta^b(C))\nonumber \\
&\oplus& (\Delta^b(D)\otimes \Delta^b(D)) \nonumber \\
&=& S^2(\Delta^b(C))\oplus \wedge^2(\Delta^b(C))\oplus (\Delta^b(C)\otimes \Delta^b(D)) \oplus (\Delta^b(D)\otimes \Delta^b(C)) \nonumber \\
&\oplus& S^2(\Delta^b(D))\oplus \wedge^2(\Delta^b(D)).\\
T^{\ast}M\otimes \Delta^b(D) &=& (\Delta^b(C)\oplus \Delta^b(D))\otimes \Delta^b(D) \nonumber \\
&=& S^2(\Delta^b(D))\oplus \wedge^2(\Delta^b(D)) \oplus (\Delta^b(C)\otimes \Delta^b(D))
\end{eqnarray}

Under this splitting $i(S^2T^{\ast}M) \subset T^{\ast}M\otimes T^{\ast}M$ has the following form 
\begin{eqnarray}
i(S^2T^{\ast}M) &=& S^2(\Delta^b(C))\oplus \mathbf{0}_{\wedge^2(\Delta^b(C))}\oplus (\Delta^b(C)\otimes \Delta^b(D)) \oplus \mathbf{0}_{(\Delta^b(D)\otimes \Delta^b(C))} \nonumber \\
&\oplus& \mathbf{0}_{\wedge^2(\Delta^b(D))}\oplus S^2(\Delta^b(D)).
\end{eqnarray}
Since $id_{T^{\ast}M}\otimes (\Delta^b \circ P_D \circ \Delta^{\sharp})$ is a full rank map the kernel and co-kernel of $\rho^a$ can be identified as 
\begin{eqnarray}
\text{ker} \rho^a &=& S^2(\Delta^b(C)) \\
\text{co-ker} \rho^a &=& \wedge^2(\Delta^b(D))
\end{eqnarray}
Following the same procedure for $\rho^b$ gives 
\begin{eqnarray}
\text{ker} \rho^b &=& S^2(\mathbf{ann}(HM)) \\
\text{co-ker} \rho^b &=& \wedge^2(\mathbf{ann}(VM))
\end{eqnarray}
Evaluating expressions (\ref{EQN:kerRho_1}) and (\ref{EQN:cokerRho_1}) using the above expressions of the kernel and co-kernel of $\rho^a$ and $\rho^b$ then proves the proposition.
\end{proof}

\begin{proposition}
The symbol $G_{1}$ of $\mathcal{R}$ is involutive.
\end{proposition}
To prove this proposition consider the following lemma.
\begin{lemma}
Let $V^{\ast}$ be a $m$-dimensional covector space, consider the subspaces $E^{\ast},F^{\ast} \subset V^{\ast}$ of dimension $r$ and $s$ respectively. If $G_1 = E^{\ast}\cup F^{\ast}$ and $G_{1+1} = S^2E^{\ast} \cap S^2F^{\ast}$, then there exists a quasi-regular basis for $V^{\ast}$.
\end{lemma}
\begin{proof}
Consider the case where $dim(E^{\ast}) > dim(F^{\ast})$ and $E^{\ast} \cap F^{\ast} \neq \lbrace \mathbf{0} \rbrace$, there exists a basis $V^{\ast} = \text{span}\lbrace v^1, \cdots , v^m \rbrace$ such that $E^{\ast}$ and $F^{\ast}$ have the following form.
\begin{eqnarray}
E^{\ast} &=& \text{span} \lbrace v^1, \cdots , v^s, v^{s+1}, \cdots , v^r \rbrace \\
F^{\ast} &=& \text{span} \lbrace v^1,\cdots , v^s \rbrace
\end{eqnarray}
The spaces $E^{\ast}\cap F^{\ast}$, $S^2(E^{\ast})$ and $S^2(F^{\ast})$ will then have the form,
\begin{eqnarray}
E^{\ast}\cap F^{\ast} &=& \text{span} \lbrace v^1,\cdots , v^s \rbrace \\
S^2(E^{\ast}) &=& \text{span} \lbrace v^{i_1} \otimes v^{i_2}\rbrace, i_1 \geq i_2.\; i_1, i_2 = 1, \cdots , r \\
S^2(F^{\ast}) &=& \text{span} \lbrace v^{j_1}\otimes v^{j_2} \rbrace, j_1 \geq j_2.\; j_1, j_2 = 1, \cdots , s \\
S^2(E^{\ast}) \cap S^2(F^{\ast}) &=& \text{span} \lbrace v^{j_1}\otimes v^{j_2} \rbrace, j_1 \geq j_2.\;  j_1, j_2 = 1, \cdots , s
\end{eqnarray}
Therefore $\text{dim}(E^{\ast}\cap F^{\ast}) = s$ and $\text{dim}((S^2(E^{\ast})) \cap S^2(F^{\ast})) = \frac{s(s+1)}{2}$. Let $\Sigma_k = \text{span}\lbrace v^{k+1}, \cdots , v^m \rbrace $ then we have, 
\begin{eqnarray}
\text{dim}((E^{\ast}\cap F^{\ast})\cap \Sigma_{k_1}) &=& s - k_1, \;\text{for} \;k_1 = 1, \cdots, s-1 \\
\text{dim}((E^{\ast}\cap F^{\ast})\cap \Sigma_{k_1}) &=& 0, \text{otherwise}.
\end{eqnarray}
From which we evaluate the following,
\begin{equation}
\text{dim}(E^{\ast}\cap F^{\ast}) + \sum_{k = 1}^{m-1}\text{dim}((E^{\ast}\cap F^{\ast})\cap \Sigma_{k}) = \frac{s(s+1)}{2}.
\end{equation}
Proving that the chosen basis is indeed quasi-regular. Following the same procedure for the cases where $dim(E^{\ast}) > dim(F^{\ast})$ and $dim(E^{\ast}) = dim(F^{\ast})$ proves the lemma.
\end{proof}

Applying the above lemma to the case where $V^{\ast} = T^{\ast}_{\mathbf{x}}M, E^{\ast} = \Delta^b(C(\mathbf{x})), F^{\ast} = \mathbf{ann}(H_{\mathbf{x}}M)$ and theorem (\ref{Thrm:InvolutiveThrm}) proves the proposition\cite{Guill}.
\subsection{Curvature map}
\label{subsec:CurvMap}
Let $p \in \mathcal{R} \subset J^1\pi$ with co-ordinates $p = (x^i, V, V_i)$, any point $q \in J^2\pi$ that projects onto $p$ will have coordinates of the form $q = (x^i, V, V_i, \tilde{V}_{[i_1, i_2]})$. The curvature map $\kappa : \mathcal{R} \mapsto \text{co-ker} \rho_1(\sigma(\Psi))$ is calculated as 

\begin{equation}
\kappa(p) = \tau \left( \rho_1(\Psi)(q) - j^1\Psi(p)\right).
\end{equation}

Let $(x^i, \mathcal{X}^a, \omega_q, \mathcal{X}^a_i, \omega_q^i)$ be local coordinates for $J^1\tilde{\pi}$. By definition $\rho_1(\Psi)(q) \in J^1\tilde{\pi}$ projects onto $\Psi(p)$ and therefore $\rho_1(\Psi)(q)$ will differ from $j^1\Psi(p)$ in the $\mathcal{X}^p_i$ and $\omega_q^i$ co-ordinates only. Furthermore since the fibre bundle $\tilde{\pi}^1_0$ has an affine structure modelled on $T^{\ast}M\otimes V\tilde{\pi} \equiv (T^{\ast}M\otimes D)\oplus (T^{\ast}M\otimes \mathbf{ann}(VM))$, 
$\rho_1(\Psi)(q) - j^1\Psi(p)$ can be taken to be an element of $(T^{\ast}M\otimes D)\oplus (T^{\ast}M\otimes \mathbf{ann}(VM))$. Applying equation (\ref{EQN:PdeCoordEqn}) and proposition (\ref{Prop:ProlongEqn}) gives, 

\begin{eqnarray}
\rho_1(\Psi)(q) - j^1\Psi(p) &=& \left[ \sum_{i_1 = 1}^m P^a_{D,i_1}\left(\tilde{V}_{[i,i_1]} -\frac{\partial }{\partial x^i}( V_{i_1}) \right) dx^i\otimes e_a \right.,\nonumber \\
& & \left. \sum_{i_1 = 1}^m P^{i_1}_{VM,q} \left( \tilde{V}_{[i,i_1]} - \frac{\partial }{\partial x^i}(V_{i_1}) \right)dx^i \otimes dx^q.\right]
\end{eqnarray}
 Recall that it is shown in proposition (\ref{Prop:CokerDefn}) that $\text{co-ker} \rho_1(\sigma(\Psi)) = \wedge^2(\Delta^b(D))\oplus \wedge^2(\mathbf{ann}(VM))$. Making the following identifications, 
\begin{enumerate}
\item $S^2T^{\ast}_{\mathbf{x}}M$ is the space of symmetric $m\times m$ matrices.
\item $T^{\ast}_{\mathbf{x}}M\otimes D(\mathbf{x})$ is the space of $(m-r)\times m$ matrices.
\item $T^{\ast}_{\mathbf{x}}M\otimes \mathbf{ann}(V_{\mathbf{x}}M)$ is the space of $m\times n$ matrices.
\item $\wedge^2(\Delta^b(D(\mathbf{x})))$ is the space of skew-symmetric $(m-r)\times (m-r)$ matrices.
\item $\wedge^2(\mathbf{ann}(V_{\mathbf{x}}M))$ is the space of skew-symmetric $n\times n$ matrices.
\end{enumerate} 
With these identifications the prolonged symbol map $\rho_1(\sigma(\Psi)) : S^2T^{\ast}M \mapsto (T^{\ast}M\otimes D)\oplus (T^{\ast}M\otimes \mathbf{ann}(VM))$ becomes $\rho_1(\sigma(\Psi))(\Omega) = (P_D\Omega^T, \Omega P_{VM})$. The map $\tau : (T^{\ast}M\otimes D)\oplus (T^{\ast}M\otimes \mathbf{ann}(VM)) \mapsto \wedge^2(\Delta^b(D))\oplus \wedge^2(\mathbf{ann}(VM))$ becomes $\tau(A,B) = (AP^T_D - P_DA^T, B^TP_{VM}-P^T_{VM}B)$.

Evaluating the curvature map $\kappa(p) = (G(p), H(p)), G(p) \in \wedge^2(\Delta^b(D(\mathbf{x}))), H(p) \in \wedge^2(\mathbf{ann}(V_{\mathbf{x}})M))$,
\begin{eqnarray}
G^{a_1,a_2}(p) &=& \sum_{i=1}^m \sum_{i_1 = 1}^m \left( P^{a_1}_{D,i_1}P^{a_2}_{D,i} - P^{a_1}_{D,i}P^{a_2}_{D,i_1}\right)\frac{\partial}{\partial x^i}(V_{i_1}). \label{EQN:MidIntegrability1}\\
H_{q_1, q_2}(p) &=& \sum_{i_1 = 1}^m\sum_{i=1}^m \left(P^i_{VM,q_1}P^{i_1}_{VM,q_2} - P^i_{VM, q_2}P^{i_1}_{VM, q_1} \right) \frac{\partial }{\partial x^i}(V_{i_1}) \label{EQN:MidIntegrability2}.
\end{eqnarray}
Since $p \in \mathcal{R}$ i.e $\Psi$ evaluates to zero at $p$, $V_{i_1}$ satisfies the following equations.
\begin{eqnarray}
\sum_{i_1 = 1}^m \left( P^a_{D,i_1}V_{i_1} - P^a_{D, i_1}X^{i_1} \right) &=& 0 \label{EQN:PDESYS}.\\
\sum_{i_1 = 1}^m P_{VM,q}^{i_1}V_{i_1} &=& 0 \label{EQN:PDECONN}.
\end{eqnarray}
Differentiating these equations with respect to $x^i$ allows for the elimination of $\frac{\partial }{\partial x^i}(V_{i_1})$ in equations (\ref{EQN:MidIntegrability1}) and (\ref{EQN:MidIntegrability2}). Differentiating equation (\ref{EQN:PDESYS}) with respect to $x^i$,
\begin{equation}
\sum_{i_1 = 1}^m\left( \frac{\partial }{\partial x^i}\left(P^a_{D, i_1}\right)\left( V_{i_1} - X^{i_1}\right) + P^a_{D,i_1}\frac{\partial }{\partial x^i}(V_{i_1}) - P^a_{D,i_1}\frac{\partial }{\partial x^i}(X^{i_1}) \right) = 0. 
\end{equation}
Substituting into (\ref{EQN:MidIntegrability1}) gives,
\begin{eqnarray}
G^{a_1,a_2}(p) &=&  \sum_{i=1}^m \sum_{i_1=1}^m \left( \left[ P^{a_1}_{D,i}\frac{\partial }{\partial x^i}(P^{a_2}_{D,i_1}) - P^{a_2}_{D,i}\frac{\partial }{\partial x^i}(P^{a_1}_{D, i_1}) \right](V_{i_1} - X^{i_1}) \right. \nonumber \\
&+& \left. \left[ P^{a_2}_{D,i}P^{a_1}_{D,i_1} - P^{a_1}_{D,i}P^{a_2}_{D, i_1}\right]\frac{\partial }{\partial x^i}(X^{i_1}) \right).
\end{eqnarray}
Recall for the partial differential equation $\mathcal{R}$ to be integrable the curvature map must be a zero map, $G^{a_1, a_2}(p)$ is a zero map if the following conditions are satisfied.
\begin{eqnarray}
\sum_{i=1}^m\left[ P^{a_1}_{D,i}\frac{\partial }{\partial x^i}(P^{a_2}_{D,i_1}) - P^{a_2}_{D,i}\frac{\partial }{\partial x^i}(P^{a_1}_{D,i_1})\right] &=& 0 \label{eqn:IntCond1}\\
\sum_{i=1}^m\sum_{i_1 = 1}^m \left[ P^{a_2}_{D,i}P^{a_1}_{D, i_1} - P^{a_1}_{D,i}P^{a_2}_{D,i_1} \right] \frac{\partial }{\partial x^i}(X^{i_1}) &=& 0 \label{eqn:IntCond2}.
\end{eqnarray}
Equation (\ref{eqn:IntCond1}) and (\ref{eqn:IntCond2}) correspond to conditions 3 and 4 in the theorem \ref{Thrm:MainTheorem}.

Differentiating equation (\ref{EQN:PDECONN}) with respect to $x^i$,
\begin{equation}
\sum_{i_1=1}^m \left[ \frac{\partial }{\partial x^i}(P^{i_1}_{VM, q_2})V_{i_1} + P^{i_1}_{VM, q_2}\frac{\partial }{\partial x^i}(V_{i_1}) \right] = 0.
\end{equation}
Substituting this into equation (\ref{EQN:MidIntegrability2}) gives
\begin{equation}
H_{q_1, q_2}(p) = \sum_{i_1 = 1}^m \left( \sum_{i=1}^m\left[ P^i_{VM,q_2}\frac{\partial }{\partial x^i}(P^{i_1}_{VM, q_1}) - P^i_{VM,q_1}\frac{\partial }{\partial x^i}(P^{i_1}_{VM, q_2})\right] \right) V_{i_1}. 
\end{equation}
Requiring $H_{q_1,q_2}(p)$ to be a zero map imposes the following condition
\begin{equation}
\label{EQN:IntegCondConn1}
\sum_{i=1}^m\left[ P^i_{VM,q_2}\frac{\partial }{\partial x^i}(P^{i_1}_{VM, q_1}) - P^i_{VM,q_1}\frac{\partial }{\partial x^i}(P^{i_1}_{VM, q_2})\right] = 0,
\end{equation}
From proposition (\ref{Prop:CoonectionComp}) the projection map $P_{VM} : T^{\ast}M \mapsto \mathbf{ann}(VM)$ in matrix form is, 
\begin{equation}
P_{VM} = 
\bbordermatrix{~ & 1 &        & n \cr
							 1 & 1 & \cdots & 0 \cr
							 \vdots & \vdots & \ddots & \vdots \cr
							 n & 0 & \cdots & 1 \cr
							 n+1 & \Gamma^{n+1}_1 & \cdots & \Gamma^{n+1}_n \cr
							 \vdots & \vdots & \ddots & \vdots \cr
							 m & \Gamma^{m}_1 & \cdots & \Gamma^{m}_{n} \cr}.
\end{equation}

The component $P^i_{VM,q}$ corresponds to the matrix element in the $i^{th}$-row and $q^{th}$ column. Substituting this into equation (\ref{EQN:IntegCondConn1}), the integrability condition becomes

\begin{equation}
\frac{\partial }{\partial x^{q_2}}\left(\Gamma^{l}_{q_1}\right) - \frac{\partial }{\partial x^{q_1}}\left(\Gamma^{l}_{q_2}\right) + \sum_{i_1 = n+1}^m\left[\Gamma^{l_1}_{q_2}\frac{\partial }{\partial x^{l_1}}\left( \Gamma^l_{q_1}\right) - \Gamma^{l_1}_{q_1}\frac{\partial }{\partial x^{l_1}}\left(\Gamma^l_{q_2}\right) \right] = 0
\end{equation} 

where $l,l_1 = n+1, \cdots , m$. This corresponds to the components of the curvature form of the connection being zero thus proving the flatness requirement in the theorem.
\begin{Remark}
The conditions stated in the main theorem (\ref{Thrm:MainTheorem}) are nothing but just the integrability conditions for the system of partial differential equations defined in proposition (\ref{Prop:PDEDefn}). If these conditions are met the function $V(\mathbf{x})$ can be constructed iteratively by solving the system of partial differential equations (\ref{Prop:PDEDefn}) via the Taylor series method. Having constructed $V(\mathbf{x})$ the stabilizing feedback controller can be found by  solving the following under-determined system of algebraic equations where $u_1, \cdots , u_r$ are the unknowns.

\begin{equation}
\begin{bmatrix}
f^1_1(\mathbf{x}) & \cdots & f^1_r(\mathbf{x})\\
\vdots & \ddots & \vdots \\
f^m_1(\mathbf{x}) & \cdots & f^m_r(\mathbf{x})
\end{bmatrix}
\begin{bmatrix}
u_1 \\
\vdots \\
u_r
\end{bmatrix} = 
\begin{bmatrix}
\tilde{g}^1(\mathbf{x}^1) - \frac{\partial \tilde{V}}{\partial x^1}(\mathbf{x}^1) - \frac{\partial V}{\partial x^1}(\mathbf{x}) \\
\vdots \\
\tilde{g}^n(\mathbf{x}^1) - \frac{\partial \tilde{V}}{\partial x^n}(\mathbf{x}^1) - \frac{\partial V}{\partial x^n}(\mathbf{x}) \\
\Gamma^{n+1}_1(\mathbf{x})\tilde{g}^1(\mathbf{x}^1) + \cdots + \Gamma^{n+1}_n(\mathbf{x})\tilde{g}^n(\mathbf{x}^1) - \frac{\partial V}{\partial x^{n+1}} \\
\vdots \\
\Gamma^m_1(\mathbf{x})\tilde{g}^1(\mathbf{x}^1) + \cdots + \Gamma^m_n(\mathbf{x})\tilde{g}^n - \frac{\partial V}{\partial x^m}
\end{bmatrix}
\end{equation}
\end{Remark}
\section{Conclusion}
\label{sec:Conc}
This paper has addressed the issue of stabilizability preserving quotients by proposing a method of constructing a control Lyapunov function of the system if it admits a stabilizable quotient system. More importantly we prove a theorem that identifies system structural obstructions to the proposed Lyapunov function construction method. Our approach to stabilizability preserving quotients focussed on lifting the stabilizability property from the lower dimensional quotient to the original system, a full characterization of stabilizability preserving quotients will also require a study of how quotients propagate the stabilizability property. A direction of futurework will therefore involve characterizing conditions under which the stabilizability property is propagated. 

\bibliographystyle{siamplain}
\bibliography{references}
\end{document}